\documentclass[sn-standardnature]{sn-jnl}

\usepackage{amsmath}
\usepackage[algo2e,ruled]{algorithm2e}
\DeclareMathOperator*{\argmax}{arg\,max}
\DeclareMathOperator*{\argmin}{arg\,min}

\jyear{2022}%

\theoremstyle{thmstyleone}%
\newtheorem{theorem}{Theorem}
\newtheorem{lemma}[theorem]{Lemma}

\theoremstyle{thmstyletwo}%

\theoremstyle{thmstylethree}%
\newtheorem{assumption}{Assumption}
\begin{document}

\title[Cooperative Actor-Critic via TD Error Aggregation]{Cooperative Actor-Critic via TD Error Aggregation}


\author*[1]{\fnm{Martin} \sur{Figura}}\email{figuramartin1@gmail.com}

\author[2]{\fnm{Yixuan} \sur{Lin}}

\author[3]{\fnm{Ji} \sur{Liu}}

\author[1]{\fnm{Vijay} \sur{Gupta}}

\affil[1]{\orgdiv{Department of Electrical Engineering}, \orgname{University of Notre Dame}}

\affil[2]{\orgdiv{Department of Applied Math and Statistics}, \orgname{Stony Brook University}}

\affil[3]{\orgdiv{Department of Electrical and Computer Engineering}, \orgname{Stony Brook University}}


\abstract{In decentralized cooperative multi-agent reinforcement learning, agents can aggregate information from one another to learn policies that maximize a team-average objective function. Despite the willingness to cooperate with others, the individual agents may find direct sharing of information about their local state, reward, and value function undesirable due to privacy issues. In this work, we introduce a decentralized actor-critic algorithm with TD error aggregation that does not violate privacy issues and assumes that communication channels are subject to time delays and packet dropouts. The cost we pay for making such weak assumptions is an increased communication burden for every agent as measured by the dimension of the transmitted data. Interestingly, the communication burden is only quadratic in the graph size, which renders the algorithm applicable in large networks. We provide a convergence analysis under diminishing step size to verify that the agents maximize the team-average objective function.}




\maketitle

\section{Introduction}\label{sec:introduction}
In multi-agent reinforcement learning (MARL), a network of agents interacts in a common environment \cite{littman1994}. Each agent seeks to optimize an objective function that is aligned with its interests. This paper contributes to the growing literature of cooperative MARL \cite{hernandez2019survey,lee2020survey} which assumes that all agents optimize a common objective function. Cooperative MARL has potential applications in power grid \cite{xi2018}, traffic networks \cite{bazzan2009}, robotics \cite{perrusquia2021}, and games \cite{vinyals2019}.\par
In cooperative MARL, the agents wish to train local policies that are independent of each other and conditioned only on the current state. This ensures that the decentralized policies can be readily executed at test time once the training is complete. While the policies have to be decentralized at test time, the training process itself can be both centralized and decentralized. The centralized paradigm has been dominant in the literature and led to empirical success in complex simulated environments \cite{castellini2021,lowe2017,foerster2018,rashid2018}. Centralized training, where the data from all agents are used to improve the collective decision-making, appears to be well-suited for simulated environments since every training episode of a multi-agent Markov decision process (MMDP) is typically assumed to run on a single simulator \cite{zhang2021survey}. However, centralized training is generally infeasible in real-world environments, where one has to account for the lack of observability, noisy communication between agents, and their lack of desire to share detailed information about their learning process including the local state, reward, and value function. Therefore, decentralized training emerged as an exciting possibility to optimize the policies of agents that have little knowledge about the global training process.

\subsection{Related work}
Decentralized cooperative MARL has lately received attention in the machine learning and control community. New MARL algorithms have been proposed for decentralized training under different assumptions about the knowledge bestowed upon the participating agents. In Table~\ref{table:1}, we present a sample of references with an emphasis on the observability of state and action and the communication range of the networked agents. Those assumptions are crucial in decentralized learning. In addition to the assumptions, we highlight the convergence properties of the listed MARL methods. When the global state and action are observable and agents receive local rewards, consensus-based MARL algorithms ensure that the agents maximize the team-average returns without direct reward-sharing \cite{kar2013,zhang2018,zhang2019policy,figura2021,suttle2020,lin2019}. The cooperation between the agents in consensus-based MARL is attained if the agents follow a simple paradigm, i.e., if they {\em update} local estimates of the team-average value function, {\em communicate} the local estimates to their neighbors, and {\em aggregate} the received estimates in consensus updates. Despite the popularity of consensus-based MARL algorithms in academic research, it may be difficult to identify a set of practical problems where the global state is observable by every agent. For problems where only decentralized training is feasible, it is natural to assume that each agent observes a local (belief) state and action. In this setting, simple independent learning \cite{tan1993,claus1998} does not generally facilitate cooperation, and thus the agents must establish a communication channel that enables them to coordinate their policies. In \cite{qu2020,qu2020avg}, the communication of the local critic values is assumed to take place between $k$-hop neighbors, and hence each agent updates its policy to maximize the expected returns (or average rewards) of the agents in the $k$-hop neighborhood. Unless the $k$-hop neighborhood includes all agents or the MMDP is decoupled, the entire network of agents is never fully cooperative. Thus, the policy converges merely to a compact set that includes a locally optimal policy. It is important to note that learning in this fashion is scalable even though the network may not exhibit desired performance in some problems. Then, the key question is whether the strengths of this approach can be used to establish a fully cooperative algorithm under mild assumptions on the communication between agents.

\begin{table}[t]
\begin{center}
\caption{Comparison of MARL methods for decentralized training}
\begin{tabular}{ @{} c c c c @{} } 
 \toprule
State observability & Communication range & Policy convergence & References \\ 
 \midrule
 global & immediate neighborhood & locally optimal & \cite{kar2013,zhang2018} \\
 \midrule
 $k$-hop neighborhood & $k$-hop neighborhood & compact set &  \cite{qu2020,qu2020avg} \\
 \midrule
 local & immediate neighborhood & locally optimal & our work\\
 \botrule
\end{tabular}

\label{table:1}
\end{center}
\end{table}

\subsection{Contribution}
In this work, we propose a decentralized AC algorithm with TD error aggregation for MMDPs with local state observability and decentralized rewards. Since the algorithm is designed for decentralized training, which is suited for real-world environments, we assume that the communication between agents is subject to time delays and packet dropouts. To the best of our knowledge, our proposed algorithm is the first algorithm that is provably robust against such disturbances. The algorithm is fully decentralized in the sense that the agents keep the local state, action, reward, and value function private, use arbitrary surrogate models to approximate their local actor and critic networks, and communicate on a time-varying directed graph. The cooperative nature of the agents is attributed to the communication protocol, through which the agents exchange local TD errors. We note that the communication burden of each agent using the general algorithm is only $KN$, where $K$ is the communication latency and $N$ is the number of agents. Furthermore, the cost can be reduced to $K$ for the algorithm designed for acyclic graphs. This makes the decentralized AC algorithm with TD error aggregation suitable for training large-scale networks. In the paper, we present both versions of the algorithm and provide a convergence analysis for the actor and critic updates with a diminishing step size. We prove that the team policy locally maximizes the approximated team-average objective function despite the disturbances in communication.

\section{Background}\label{sec:background}
In this section, we formally define the MARL problem. We consider a jointly observable multi-agent Markov decision process (JOMMDP) that assumes that the environment state is jointly observable but each agent on its own observes only its lower-dimensional realization (a local state). Moreover, we formulate the objective of the cooperative agents and present the policy gradient for the multi-agent AC methods that will be essential in the design of the decentralized AC algorithm with TD error aggregation in the next section.
\subsection{Jointly Observable Multi-agent Markov Decision Process}
A JOMMDP is given as a tuple $(\mathcal{S},\{\mathcal{A}^i\}_{i\in\mathcal{N}},\mathcal{P},\{\mathcal{R}^i\}_{i\in\mathcal{N}},\{\mathcal{O}^i\}_{i\in\mathcal{N}},\mathcal{G},\gamma)$, where $\mathcal{N}=\{1,\dots,N\}$ is a set of agents, $\mathcal{S} $ is a set of states, $\mathcal{P}$ is a set of transitional probabilities, $\mathcal{G}$ represents a set of graphs, and $\gamma\in[0,1)$ is a discount factor. Furthermore, $\mathcal{A}^i$, $\mathcal{R}^i$, and $\mathcal{O}^i$ are a set of actions, rewards, and states observed by agent~$i$, respectively. A time-varying graph $\mathcal{G}_t=(\mathcal{N},\mathcal{E}_t)\in\mathcal{G}$ is defined by a set of vertices $\mathcal{N}$ associated with the agents in the network and a set of edges $\mathcal{E}_t\subseteq\mathcal{N}\times\mathcal{N}$. We denote the state observed and action taken by agent~$i$ as $s^i$ and $a^i$, respectively, and write the global state and action in the vector form $s=[(s^1)^T,\dots,(s^N)^T]^T$ and $a=[(a^1)^T,\dots,(a^N)^T]^T$, respectively. We let $s^\prime$ denote the global state after a transition from $s$. The reward of agent~$i$ is given as $r^i(s,a,s^\prime):\mathcal{S}\times\mathcal{A}\times\mathcal{S}\rightarrow\mathcal{R}^i\subset\mathbb{R}$ and $p(s^\prime\vert s,a):\mathcal{S}\times\mathcal{S}\times\mathcal{A}\rightarrow\mathcal{P}\subset\mathbb{R}$ denotes the joint transitional probability. We note that the definition of the JOMMDP dynamics and rewards is general since it allows for arbitrarily strong coupling between the agents. We let $\pi^i(a^i\vert s^i):\mathcal{S}^i\times\mathcal{A}^i\rightarrow(0,1)$ denote the policy of agent~$i$. The global policy $\pi(a\vert s)=\prod_{i\in\mathcal{N}}\pi^i(a^i\vert s^i)$ induces a stationary distribution of the Markov chain associated with the JOMMDP, denoted as $d_\pi(s)$. Moreover, we let $p_\pi(s^\prime\vert s)=\sum_{a\in\mathcal{A}}p(s^\prime\vert s,a)\pi(a\vert s)$ denote the joint state transition probability induced by policy $\pi(a \vert s)$. Throughout the paper, we frequently use subscript $t$ to emphasize the dependency of signals on time, e.g., $r_{t+1}^i(s_t,a_t,s_{t+1})$. Every agent is aware only of its own action $a^i$ in addition to the local state observation $s^i$ and local reward $r^i$. Furthermore, agent~$i$ receives information from a time-varying set of agents denoted as $\mathcal{N}_{in,t}^i$ and transmits information to a time-varying set of agents denoted as $\mathcal{N}_{out,t}^i$. Finally, we make a following assumption about the communication between agents.
\begin{assumption}\label{as:communication}
A message transmitted by agent~$i$ at $t_s\in[t,t+T_1]$, for $t\geq 0$, is received by a neighboring agent $j\in\mathcal{N}_{out,t}^i$ at time $t_r\in[t_s,t_s+T_2]$.  The length of an interval, in which at least one message is successfully transmitted, $T_1$, and the maximum delay in the transmission, $T_2$, are uniformly bounded. We also assume that a message reaches every agent in the network in at most $k$ hops, where $k$ is a finite positive number. Therefore, the latency from agent~$i$ to agent $j$ is uniformly bounded by a constant $K=k\cdot(T_1+T_2)$.
\end{assumption}

\subsection{Objective}
We are ready to state the objective for cooperative agents that apply TD learning with a one-step look-ahead. Before training, each agent identifies a nominal local objective function $
J^i(\pi)=\mathbb{E}_{d_\pi}\big[V_\pi^i(s)\big]$, where $V_\pi^i(s)=\mathbb{E}_{\pi,\mathcal{P}}\big[r^i(s,a,s^\prime)+\gamma V_\pi^i(s^\prime)\vert s\big]$ is the value function. The global objective function is given as a simple average over all cooperative agents, i.e., $J(\pi)=\frac{1}{N}\sum_{i=1}^NJ^i(\pi)$. Along with this formulation, we define the team-average value function $V_\pi(s)=\frac{1}{N}\sum_{i\in\mathcal{N}}V_\pi^i(s)$ and the team-average action value function $Q_\pi(s,a)=\mathbb{E}_{\mathcal{P}}\big[\frac{1}{N}\sum_{i\in\mathcal{N}}r^i(s,a,s^\prime)+\gamma V_\pi(s^\prime)\vert s,a\big]$. It is important to note that the agents do not know the nominal local objective functions of one another, yet they collectively agree to maximize the team-average objective function $J(\pi)$. This is formally stated in the following optimization problem:
\begin{align}\label{local opt problem}
\pi_*^i=\argmax_{\pi^i} J(\pi),\quad \forall i\in\mathcal{N}.
\end{align}
We stress that all agents optimize over their local policies $\pi^i$ even though the objective function $J(\pi)$ depends on the global policy $\pi$. This is a unique challenge that typifies MARL, where an agent views all other agents as part of the environment \cite{zhang2021survey}. Our goal is to approximately solve the optimization problem in \eqref{local opt problem} using an AC method, which requires sampling of the baseline policy gradient $\nabla_\pi J(\pi)$ given as follows \cite{sutton2018book}:
\begin{align}
\nabla_\pi J(\pi)=\mathbb{E}_{\pi,d_\pi}\big[\big(Q_\pi(s,a)-V_\pi(s)\big)\nabla_\pi\log\pi(a\vert s)\big].\label{policy_gradient}
\end{align}
By the conditional independence of the individual policies $\pi^i(a^i\vert s^i)$, the local policy gradient is given as
\begin{align}
\nabla_{\pi^i} J(\pi)=\mathbb{E}_{\pi,d_\pi}\big[\big(Q_\pi(s,a)-V_\pi(s)\big)\nabla_{\pi^i}\log\pi^i(a^i \vert  s^i)\big].\label{local_policy_gradient}
\end{align}
We note that the evaluation of the local policy gradient in \eqref{local_policy_gradient} is nontrivial. The main challenge lies in the fact that the agents cannot directly take samples of the team-average advantage function, referred to as the TD error, as follows:
\begin{align}
\sum_{i\in\mathcal{N}}\big(r^i(s,a,s^\prime)+\gamma V_\pi^i(s^\prime)-V^i_\pi(s)\big)\sim Q_\pi(s,a)-V_\pi(s).\label{TD_error_agg}
\end{align}
This is due to the private and spatially distributed local signals. A further challenge lies in the estimation of the local value function $V_\pi^i(s)$, which must be based on the lower-dimensional local observation $s^i$ since the global state $s$ is not observable. In the next section, we introduce a multi-agent AC method that involves local estimation, communication on time-varying directed graphs, and TD error aggregation that facilitates approximate sampling of the local policy gradient in \eqref{local_policy_gradient}. The method is designed for real-world environments, where the agents wish to retain privacy and the communication between them is subject to time delays and packet dropouts.

\section{Methodology}\label{sec:methodology}
In this section, we present a decentralized AC method with TD error aggregation that addresses the challenges presented in Section~\ref{sec:background}. It involves the transmission of TD errors over the communication graph as well as the aggregation of the TD errors by every agent. The decentralized AC method with TD error aggregation assumes delayed, yet synchronized actor updates. In addition to the general online algorithm in Section~\ref{subsec:DAC-TD}, we also present its communication-efficient version designed for acyclic communication graphs in Section~\ref{subsec:DAC-TD2}.

\subsection{Decentralized actor-critic method with TD error aggregation}\label{subsec:DAC-TD}
In this section, we introduce a method that would allow the agents to find a locally optimal cooperative policy. In Section~\ref{sec:background}, particularly in \eqref{policy_gradient} and \eqref{TD_error_agg}, we established that the true team-average policy gradient includes the local value function $V_\pi^i(s)$ which we refer to as the critic. It is important to note that the critic $V_\pi^i(s)$ is a function of the global state~$s$ that is generally not observable by agent~$i$. Furthermore, we assume that the state space $\mathcal{S}$ is large, and thus storing values of $V_\pi^i(s)$ for all $s\in\mathcal{S}$ is not admissible. Therefore, we consider an approximation of $V_\pi^i(s)$ via a parametric model $V^i(s^i;v^i)$, where $v^i$ are the model parameters. Such approximation is scalable in the sense that the dimensionality of $s^i$ and $v^i$ does not grow with the number of agents in the network. Hence, the model complexity does not grow with the size of the network either. We apply the same arguments to introduce a parametric model $\pi^i(a^i\vert s^i;\theta^i)$ that approximates the policy (actor). We write the global policy as $\pi_\theta=\pi(a\vert s;\theta)=\prod_{i\in\mathcal{N}}\pi^i(a^i\vert s^i;\theta^i)$, where $\theta=[(\theta^1)^T,\dots,(\theta^N)^T]^T$.\par
While the ultimate goal of the cooperative agents is to learn a cooperative policy through the iteration
\begin{align}
    \theta_{t+1}^i=\theta_t^i+\alpha_t\cdot\tilde\nabla_{\theta^i}J(\pi_{\theta_t}),\label{actor_update}
\end{align}
where $\alpha_t$ is the actor learning rate, the actor and critic networks are intimately connected in the training process as the agents need to estimate the policy gradient before they apply it in the local actor update. The process of learning the critic $V^i(s^i;v^i)$ is referred to as the policy evaluation. In the policy evaluation, each agent solves a local optimization problem given as follows:
\begin{align}
v^i_\pi=\argmin_{v^i}\mathbb{E}_{d_\pi}\big[\mathbb{E}_{\pi,\mathcal{P}}\big(r^i(s,a,s^\prime)+\gamma V^i(s^{i\prime};v^i)-V^i(s^i;v^i)\big)^2\big].\label{policy_evaluation}
\end{align}
Solving this mean squared error estimation problem is straightforward because all local signals are available to agent~$i$. Having defined an optimization problem for the policy evaluation, we move onto the central question in our work. \textit{How do we ensure that the agents cooperate?}\par
Our solution is relatively simple and easily fits into the setting of decentralized networks with communication delays. We propose that the agents aggregate TD errors $\delta^i=r^i(s,a,s^\prime)+\gamma V^i(s^{i\prime};v^i)-V^i(s^i;v^i)$ received from their neighbors, and thus they obtain samples of the team-average advantage function $Q_\pi(s,a)-V_\pi(s)$ according to \eqref{TD_error_agg}. However, it is important to note that the actor updates defined in \eqref{actor_update} cannot be executed online since the TD errors $\delta^i$ are not propagated across the entire network instantly under Assumption~\ref{as:communication}. Therefore, we let the agents exchange $\delta^i$ and synchronously update the local actor parameters $\theta^i$ once all local TD errors have been received by every agent in the network. This leads to an online actor update in the form $\theta_{t+1}^i=\theta_t^i+\alpha_{t-K}\cdot\tilde\nabla_{\theta^i}J(\pi_{\theta_{t-K}})$, where $K$ is the maximum number of steps for data to be received by every agent in the network (see Assumption~\ref{as:communication}). The propagation of the local TD errors in a latent network is depicted in Fig.~\ref{fig:TD communication}. To successfully deliver TD errors across the network that experiences time delays and packet dropouts in communication according to Assumption~\ref{as:communication}, we need to carefully design the communication protocol. We define a vector $\Delta_{t-\tau}^i\in\mathbb{R}^N$ as a collection of local TD errors that were observed by individual agents at time $t-\tau$ and have been received by agent~$i$. The entire history of TD errors maintained by agent~$i$ entails a set of such vectors, $\{\Delta_{t-\tau}^i\}_{\tau\in[0,K]}$. We note that the history of TD errors is incomplete due to the network latency, and thus it has to be incrementally updated based on new information provided by the agent's neighbors. In the following lemma, we show that each agent can evaluate the team-average TD error with a delay as long as all agents transmit the history of TD errors $\{\Delta_{t-\tau}^i\}_{\tau\in[0,K]}$ at every time step.
\begin{figure}[t]
\centering
\includegraphics[width=1\linewidth]{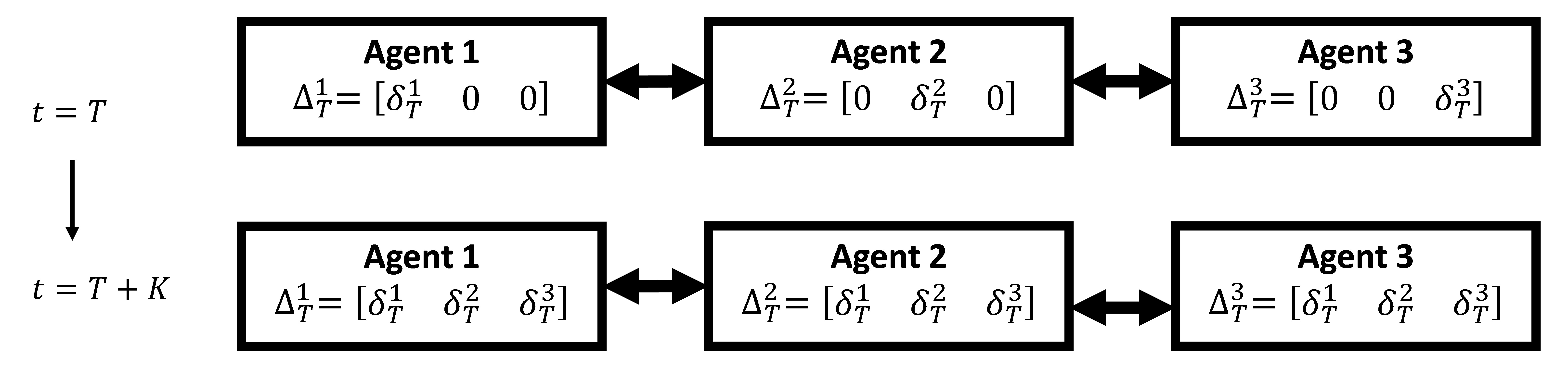}
\caption{The propagation of TD error vectors $\Delta_T$ in a simple network for $t\in[T,T+K]$.}
\label{fig:TD communication}
\end{figure}
\begin{lemma}\label{lem:TD aggregation}
Suppose that Assumption~\ref{as:communication} holds. For  $t\geq 0$, each agent initializes a vector $\Delta_t^i$ such that $[\Delta_t^i]_j=\delta_t^i$ for $i=j$ and $[\Delta_t^i]_j=0$ otherwise. Furthermore, each agent transmits a set of vectors $\{\Delta_{t-\tau}^i\}_{\tau\in[0,K-1]}$, and updates the zero entries of each vector in the set with the received histories $\{\Delta_{t-\tau}^j\}_{\tau\in[1,K]}$, $j\in\mathcal{N}_{in,t}^i$. Then, the team-average TD error satisfies $\delta_{t-K}=\frac{1}{N}\mathbf{1}^T\Delta_{t-K}^i$ for $i\in\mathcal{N}$ and $t\geq K$.
\end{lemma}
\begin{proof}
For simplicity, we consider updates of a single vector $\Delta_{t}^i$ rather than of the entire set. Agent~$i$ initializes $\Delta_{t}^i=\mathbf{0}$ and assigns $[\Delta_{t}^i]_i=\delta_{t}^i$ at time $t=T$. Then, the zero entries of $\Delta_T^i$ are updated upon the reception of $\Delta_T^j$ for $t\in[T,T+K]$. By Assumption~\ref{as:communication}, this update rule yields $[\Delta_T^i]_j=\delta_T^j$ for all $j\in\mathcal{N}$ at $t\in[T,T+K]$. Therefore, the team-average TD error can be evaluated as $\delta_T=\frac{1}{N}\mathbf{1}^T\Delta_T^i$ at time $t=T+K$, which concludes the proof.
\end{proof}
\noindent
The propagation and aggregation of TD errors validated in Lemma~\ref{lem:TD aggregation} are essential steps in the decentralized AC algorithm with TD error aggregation. A pseudo-code for the online algorithm is presented in Algorithm~\ref{alg:vanilla}. In the next subsection, we present an extension of Algorithm~\ref{alg:vanilla} for acyclic graphs that reduces the communication burden imposed on the agents. The convergence analysis of both algorithms is provided in Section~\ref{sec:convergence}.

\begin{algorithm2e}[h]
\label{alg:vanilla}
\caption{Decentralized AC algorithm with TD error aggregation}
 \textbf{Initialize parameters} $\theta_0^i,v_0^i$, $\forall i\in\mathcal{N}$\\
 \textbf{Initialize} $K$, $s_0,\{\alpha_{t-K}\}_{t\geq K},\{\beta_{t}\}_{t\geq 0}$, $\{\Delta^i_{t-\tau}\}_{\tau=0:K}\leftarrow\mathbf{0}$, $\{\eta_{t-\tau}^i\}_{\tau=0:K}\leftarrow\mathbf{0}$ $\forall i\in\mathcal{N}$, $t\leftarrow 0$\\
\textbf{Repeat until convergence} \\
 \For{$i\in\mathcal{N}$}
 	 {
 	   \textbf{Take action} $a_{t}^i\sim\pi^i(a_{t}^i \vert s_{t}^i;\theta_{t}^i)$\\
 	   \textbf{Observe} $s_{t+1}^i$ and $r_{t+1}^i$\\
 	   \textbf{Receive} $\{\Delta_{t-\tau}^j\}_{\tau=1:K}$ from $j\in\mathcal{N}_{in,t}^i$\\
  	   \textbf{Compute} $\delta_t^i\leftarrow r_{t+1}^i+\gamma V^i(s_{t+1}^i;v_t^i)-V^i(s_t^i;v_t^i)$\\
  	   \textbf{Update critic}	$v_{t+1}^i\leftarrow v_t^i+\beta_{t}\cdot\delta_t^i\cdot\nabla_{v^i} V(s_t;v_t^i)$\\	
  	   \textbf{Compute} $\eta_t^i\leftarrow\nabla_{\theta^i}\log\pi^i(a_{t}^i \vert s_{t}^i;\theta_{t}^i)$\\
  	   \textbf{Update} $\{\eta_{t-\tau}^i\}_{\tau=0:K}$\\
  	   \textbf{Initialize TD error vector} $\big[\Delta_{t}^i\big]_i\leftarrow \delta_t^i$\\
  	   \textbf{Send} $\{\Delta_{t-\tau}^i\}_{\tau=0:K-1}$ to $\mathcal{N}_{out,t}^i$\\
  	   \textbf{Update} (For $\tau=1:K$ and $n=1:N$) $\big[\Delta_{t-\tau}^i\big]_n\leftarrow \big[\Delta_{t-\tau}^j\big]_n \text{ if }\exists j\in\mathcal{N}_{in,t}^i;\, \big[\Delta_{t-k}^j\big]_n\neq 0$\\
  		\textbf{Compute} $\delta_{t-K}\leftarrow \frac{1}{N}\mathbf{1}^T\Delta_{t-K}^i$\\
  		\textbf{Update actor}	$\theta_{t+1}^i\leftarrow\theta_t^i+\alpha_{t-K}\cdot\delta_{t-K}\cdot\eta_{t-K}^i$
  		}
  	\textbf{Update} iteration counter $t\leftarrow t+1$
\end{algorithm2e}

\subsection{Extension for acyclic graphs}\label{subsec:DAC-TD2}
In this section, we introduce a special version of Algorithm~\ref{alg:vanilla} tailored for acyclic time-invariant communication graphs with fixed delays and no packet dropouts.
For simplicity, we assume that the fixed delays are equal to one time step, hence Assumption~\ref{as:communication} holds with $T_1=0$ and $T_2=1$. Our goal is to reduce the size of messages communicated between agents and perform actor and critic updates that are identical to the updates in Algorithm~\ref{alg:vanilla}. Under this algorithm, the agents are required to communicate a $K$-dimensional message to their neighbors as opposed to the $KN$-dimensional message under Algorithm~\ref{alg:vanilla}.  The extended algorithm achieves a message size independent of $N$ since the agents perform algebraic operations on the received data in an effort to evaluate the team-average TD error in finite time. This is very different from Algorithm~\ref{alg:vanilla}, where the received TD errors are directly communicated between the neighbors until they are finally aggregated by each agent after $K$ time steps.\par
We define a vector of team-average TD errors evaluated by agent~$i$, $\hat\Delta_t^i\in\mathbb{R}^K$. This vector includes entries associated with the team-average TD errors observed over the last $K$ time steps. The decentralized AC algorithm with TD error aggregation designed for acyclic graphs is given in Algorithm~\ref{alg:acyclic}. The agents update $\hat\Delta_t^i$ to retrieve the team-average TD error $\delta_{t-K}$ from $\hat\Delta_t^i$ at every time step. They use latent variables $\rho_t^i$ and $z_t^i$ to keep track of the data sent to their neighbors.

\begin{algorithm2e}[h]
\label{alg:acyclic}
\caption{Decentralized AC algorithm with TD error aggregation for acyclic graphs}
 \textbf{Initialize parameters} $\theta_0^i,v_0^i$, $\forall i\in\mathcal{N}$\\
 \textbf{Initialize} $k$, $s_0,\{\alpha_{t-K}\}_{t\geq K},\{\beta_{t}\}_{t\geq 0}$, $\{\eta_{t-\tau}^i\}_{\tau=0:K}\leftarrow\mathbf{0}$ $\forall i\in\mathcal{N}$, $t\leftarrow 0$\\
 	 $\qquad\qquad\rho_{t}^i\leftarrow \mathbf{0}$, $\hat\Delta_{t}^i\leftarrow \mathbf{0}$, $z_{t}^{ij}\leftarrow\mathbf{0}$ for $t\leq 0$ and $\quad\forall i,j\in\mathcal{N}$\\
\textbf{Repeat until convergence} \\
 \For{$i\in\mathcal{N}$}
 	{
    \textbf{Take action} $a_{t}^i\sim\pi^i(a_{t}^i \vert s_{t}^i;\theta_{t}^i)$\\
    \textbf{Observe} $s_{t+1}^i$ and $r_{t+1}^i$\\
    \textbf{Receive} $\rho_{t-1}^j$ from $j\in\mathcal{N}_{in,t}^i$\\
    \textbf{Compute} $\delta_t^i\leftarrow r_{t+1}^i+\gamma V^i(s_{t+1}^i;v_t^i)-V^i(s_t^i;v_t^i)$\\
    \textbf{Update critic}	$v_{t+1}^i\leftarrow v_t^i+\beta_{t}\cdot\delta_t^i\cdot\nabla_{v^i} V(s_t;v_t^i)$\\
    \textbf{Update} $[\hat\Delta_t^i]_1, [\rho_{t}^i]_1, [z_{t}^{ij}]_1\leftarrow \delta_t^i$\\
    \For{$\tau=1:K-1$}
      {
        $[\hat\Delta_t^i]_{\tau+1}\leftarrow [\hat\Delta_{t-1}^i]_\tau+\sum_{j\in\mathcal{N}_{in,t}^i}([\rho_{t-1}^j]_\tau-[z^{ij}_{t-1}]_\tau)$\\
        $[\rho_t^i]_{\tau+1}\leftarrow [\hat\Delta_{t}^i]_{\tau+1}-[\hat\Delta_{t-1}^i]_\tau$\\
        $[z_{t+1}^{ij}]_{\tau+1}\leftarrow [z_{t-1}^{ij}]_{\tau-1}+[\rho_{t}^i]_{\tau+1}-[\rho_{t-1}^j]_\tau$
      }
    \textbf{Send} $\rho_t^i$ to $\mathcal{N}_{out,t}^i$\\
    \textbf{Compute} $\eta_t^i\leftarrow\nabla_{\theta^i}\log\pi^i(a_{t}^i \vert s_{t}^i;\theta_{t}^i)$\\
    \textbf{Update} $\{\eta_{t-\tau}^i\}_{\tau=0:K}$\\
    \textbf{Compute} $\delta_{t-K}=[\hat\Delta^i_t]_K\cdot N^{-1}$\\
    \textbf{Update actor} $\theta_{t+1}^i\leftarrow\theta_t^i+\alpha_{t-K}\cdot\delta_{t-K}\cdot\eta_{t-K}^i$
  	}
  	\textbf{Update} iteration counter $t\leftarrow t+1$
\end{algorithm2e}

In the following lemma, we show that the updates of  Algorithm~\ref{alg:acyclic} yield the team-average TD error $\delta_{t-K}$ for $t\geq K$. The proof of the lemma is relegated to Appendix~\ref{appendix:extension}.

\begin{lemma}\label{lem:TD aggregation2}
For all $i\in\mathcal{N}$ and $t\geq K$, $[\hat\Delta_t^i]_K\cdot N^{-1}=\delta_{t-K}$.
\end{lemma}
\noindent
In the next section, we provide a convergence analysis of Algorithm~\ref{alg:vanilla} and specify the conditions under which the convergence results apply to Algorithm~\ref{alg:acyclic} as well.

\section{Convergence analysis}\label{sec:convergence}
In this subsection, we provide a convergence analysis for the actor and critic updates in Algorithm~\ref{alg:vanilla}. The purpose of the analysis is to show that the agents are fully cooperative in the sense that they maximize an approximation of the team-average objective function~$J(\pi_\theta)$. We use the two timescale principle that allows us to separate the actor and critic updates. First, we prove that the local critic parameters $v_t^i$ asymptotically converge with probability one to the minimizer of \eqref{policy_evaluation}. Second, we prove that the actor parameters $\theta_t^i$ asymptotically converge to a compact set of equilibria that locally maximize the team-average objective function $J(\pi_\theta)$. The proofs of the convergence theorems are provided in Appendix~\ref{appendix:proofs}.\par
We define matrices $P_\pi=[p_\pi(s^\prime \vert s)]_{s^\prime,s}\in\mathbb{R}^{ \vert\mathcal{S}\vert\times\vert\mathcal{S}\vert}$ and $D_\pi=[d_\pi(s)]_{s,s}\in\mathbb{R}^{\vert \mathcal{S} \vert \times \vert \mathcal{S} \vert}$ as well as vectors $\hat{R}_\pi^i=\big[\mathbb{E}_{\pi,\mathcal{P}}\big(r^i(s,a,s^\prime)\big)\big]_{s}\in\mathbb{R}^{\vert\mathcal{S}\vert}$ for $i\in\mathcal{N}$. We make the following assumptions.

\begin{assumption}\label{as:linear_approx}
The local critic of agent~$i$, $V^i(s^i;v^i)$, is approximated by a linear model, i.e., $V^i(s^i;v^i) = (v^i)^T\phi^i(s^i)$, where $\phi^i(s^i)=[\phi_1^i(s^i),\dots,\phi_{L^i}^i(s^i)]^T\in\mathbb{R}^{L^i}$ is a uniformly bounded feature vector. We define the local feature matrix $\Phi^i=[\phi_l^i(s^i)^T]_{s,l}\in\mathbb{R}^{\vert\mathcal{S}\vert\times L^i}$. The local feature matrix $\Phi^i$ has full column rank.
\end{assumption}

\begin{assumption}\label{as:policy}
The local actor of agent~$i$, $\pi^i(a^i \vert s^i;\theta^i)$ is stochastic, i.e., $\pi(a^i \vert s^i;\theta^i)>0$ for any $\theta^i\in\Theta^i$, $s^i\in\mathcal{S}^i$, $a^i\in\mathcal{A}^i$, and continuously differentiable in $\theta^i$. The Markov chain $\{s_t\}_{t\geq 0}$ is irreducible and aperiodic under any joint policy $\pi_\theta$.
\end{assumption}

\begin{assumption}\label{as:reward_bound}
The reward $r_{t+1}^i(s_t,a_t,s_{t+1})$ is uniformly bounded for any $i\in\mathcal{N}$ and $t\geq 0$.
\end{assumption}

\begin{assumption}\label{as:step_size}
The critic step size $\beta_{t}$  satisfies $\sum_t\beta_{t}=\infty$, and $\beta_{t}\rightarrow 0$,  $\sum_t \beta^2_{t}<\infty$. The actor step size $\alpha_{t}$ satisfies $\sum_t\alpha_{t}=\infty$, $\alpha_{t}\rightarrow 0$, $\sum_t \alpha^2_{t}<\infty$, and $\alpha_{t}=o(\beta_{t})$.
\end{assumption}

\begin{assumption}\label{as:policy_updates}
The update of the actor parameters $\theta^i$ includes a projection operator $\Psi_{\Theta^i}:\mathbb{R}^{m_i}\rightarrow\Theta^i\subset\mathbb{R}^{m_i}$, where $\Theta^i$ is a hyper-rectangular constraint set.
\end{assumption}

\noindent
These assumptions are standard in the RL literature \cite{zhang2018}. We assume a linear approximation of the critic and the associated feature matrix $\Phi^i$ to be full rank so as to ensure convergence to a unique asymptotically stable equilibrium in the policy evaluation \cite{tsitsiklis1999}. The second assumption admits a nonlinear approximation of the actor \cite{konda2000,bhatnagar2009}. This and the assumption that the rewards remain uniformly bounded ensure bounded actor and critic gradients. Assumption~\ref{as:step_size} is a sufficient condition for the convergence of a stochastic gradient descent type of algorithm  that we use for the actor and critic updates \cite{kushner2003book}. It also allows to apply the two-timescale principle in our convergence analysis, since the learning rate $\beta_t$ decays slower than $\alpha_t$ \cite{borkar2009book}. Finally, the actor parameter updates are subject to a constraint that prevents them from growing unbounded. This simplifies the analysis of the actor convergence. We are ready to present a convergence theorem for the policy evaluation.

\begin{theorem}\label{thm: critic}
Under Assumption~\ref{as:policy}-\ref{as:step_size} and fixed policy $\pi(a \vert s;\theta)$, the critic parameters $v_t^i$ are uniformly bounded for $t\geq0$ and converge to a fixed point $v_\pi^i$ with probability one for all $i\in\mathcal{N}$. The fixed point $v_\pi^i$ is a unique asymptotically stable equilibrium of the ODE
\begin{align}\label{critic ODE}
\dot{v}^i=\Phi^{iT}D_\pi(\gamma P_\pi-I)\Phi^i v^i+\Phi^{iT}D_\pi\hat{R}_\pi^i.
\end{align}
\end{theorem}
Using the fact that the local critic parameters $v^i$ converge on the faster timescale by Theorem~\ref{thm: critic}, we let
\begin{align}
\delta_{\pi}=\frac{1}{N}\sum_{i\in\mathcal{N}}\big(r^i(s,a,s^\prime)+\gamma V^i(s^{i\prime};v_\pi^i)-V^j(s^i;v_\pi^i)\big)
\end{align}
denote the team-average TD error upon the critic convergence. What follows next is the main theorem that asserts policy convergence on the slower timescale.
\begin{theorem}\label{thm:actor}
Under Assumption~\ref{as:communication}-\ref{as:policy_updates}, the policy parameter $\theta_t^i$ asymptotically converges with probability one to a compact set of stationary points of the ODE
\begin{align}
\dot\theta^i&=\Psi^i_\Theta\big[\mathbb{E}_{d_\pi,\pi,\mathcal{P}}\big(\delta_\pi\cdot\nabla_{\theta^i}\log\pi_\theta^i(a^i \vert s^i;\theta^i\big)\big].\label{actor ODE}
\end{align}
\end{theorem}

\noindent
It is important to note that the policies of agents that employ the decentralized AC with TD error aggregation locally maximize only an approximation of the objective function $J(\pi_\theta)$. The local maxima of the approximate objective function, which we denote $J^\prime(\pi_\theta)$, are the stationary points of the ODE in \eqref{actor ODE} and do not coincide with the maxima of the original objective function $J(\pi_\theta)$, because the team-average TD error $\delta_\pi$ involves an approximation of the true value function. The discrepancy in the gradients is captured through the following relationship:

\begin{align*}
&\nabla_\theta{J}^\prime(\pi_\theta)
=\nabla_\theta J(\pi_\theta)+\nabla_\theta J^\prime(\pi_\theta)-\nabla_\theta J(\pi_\theta)\\
&=\nabla_\theta J(\pi_\theta)+\mathbb{E}_{d_\pi,\pi,\mathcal{P}}\big[\sum_{i\in\mathcal{N}}\big(r^i(s,a,s^\prime)+\gamma V^i(s^{i\prime};v_\pi^i)-V(s^i;v_\pi^i)\big)\nabla_{\theta}\log\pi_\theta\big]\\
&\qquad -\mathbb{E}_{d_\pi,\pi,\mathcal{P}}\big[\sum_{i\in\mathcal{N}}\big(r^i(s,a,s^\prime)+\gamma V_\pi^i(s^\prime)-V^i_\pi(s)\big)\nabla_{\theta}\log\pi_\theta\big]\\
&=\nabla_\theta J(\pi_\theta)\\
&\quad+\mathbb{E}_{d_\pi,\pi,\mathcal{P}}\big[\sum_{i\in\mathcal{N}}\big(\gamma V^i(s^{i\prime};v_\pi^i)-\gamma V_\pi^i(s^\prime)-V(s^i;v_\pi^i)+V^i_\pi(s)\big)\nabla_\theta\log\pi_\theta\big].
\end{align*}
From the last equality, it is evident that the estimation of the critic parameters $v^i$ and the assumption of local state $s^i$ contribute to a bias in the policy gradients. To reduce the bias, we would need to assume full state observability, which may not be feasible in the training phase due to the privacy concerns and scalability issues. Therefore, the success of our proposed cooperative AC method hinges on the ability of a practitioner to strike balance between observability and scalability. This process is somewhat similar to feature engineering in general machine learning. We do not address this fundamental challenge in this paper but want to make the reader aware of it. Finally, we note that the convergence theorems for Algorithm~\ref{alg:vanilla} extend to Algorithm~\ref{alg:acyclic}, except that we enforce $T_1=0$ and $T_2=1$ in Assumption~\ref{as:communication}. In the next section, we implement the decentralized AC algorithm with TD error aggregation in a simulated environment.

\section{Empirical analysis}
In this section, we present an empirical analysis of the decentralized AC algorithm with TD error aggregation (we refer to the algorithm as DAC-TD). We compare the performance of DAC-TD with the scalable AC (SAC) algorithm \cite{qu2020} and independent learners that employ the AC algorithm (AC) in a simple environment.\par
We consider a set of five agents, $\mathcal{N}=\{1,2,3,4,5\}$, that communicate on a line graph. Agent~$i$, $i\in\mathcal{N}$, has a binary local state and action space, $\mathcal{A}^i=\mathcal{O}^i=\{0,1\}$. The transition probabilities of agent $i\in\mathcal{N}$ and the rewards of the first agent are given as $p(s_{t+1}^i=1\vert s_t,a_t)=r^1_{t+1}(s_t,a_t)=\frac{1}{2N}\sum_{i\in\mathcal{N}}(s_t^i+a_t^i)$. All other agents in the set $\mathcal{N}$ receive zero rewards. This environment is intentionally crafted such that the local states are coupled and the optimal action for every agent is $a_t^i=1$. Note that by choosing  $a_t^i=1$, all agents are more likely to transition to $s_t^i=1$ and agent~$1$ collects a higher reward.\par
In training, we set the discount factor to $\gamma=0.9$, the learning rates for the actor and critic to $\alpha_t=0.01$ and $\beta_t=0.1$, respectively. We use a neural network with two hidden layers and 10 hidden units to model the actor and a neural network with two hidden layers and 5 hidden units to model the critic. Each hidden layer applies a leaky ReLU activation function with the parameter $0.3$. The local critics are trained over $25$ epochs, where the TD targets are re-computed every $5$ epochs. Each training episode has $100$ steps and we train the agents for $1000$ episodes. We assume that agents train their local actors and critics and transmit their TD errors at the end of each episode. By the definition of the line graph with five agents, the associated graph diameter is $K=4$. Therefore, it takes four episodes for the local TD errors to be received by every agent in the network and finally appear in the local DAC-TD actor updates. For SAC, we assume that each agent receives the TD errors from all agents in a one-hop neighborhood after one episode. For AC, no communication occurs between the agents, and hence their actor updates are not subject to delays.\par
We present the simulations results in Fig.~\ref{fig:synthetic}. The empirical result is consistent with our theory; the policy of the DAC-TD agents maximizes the team objective despite the time delays in the actor updates since the local TD errors are propagated to all agents in the network. We note that neither the AC agents nor the one-hop SAC agents achieve the team-average objective because the neighborhood size, over which the agents maximize their policy, is too small and most of the agents have no incentive to change their policy as they receive zero rewards. We note that the performance of DAC-TD is the same as of $K$-hop SAC that is omitted from the figure. This is because DAC-TD can be viewed as a generalization of the $K$-hop SAC for networks with communication latency.
\begin{figure}[t]
\centering
\includegraphics[width=0.9\linewidth]{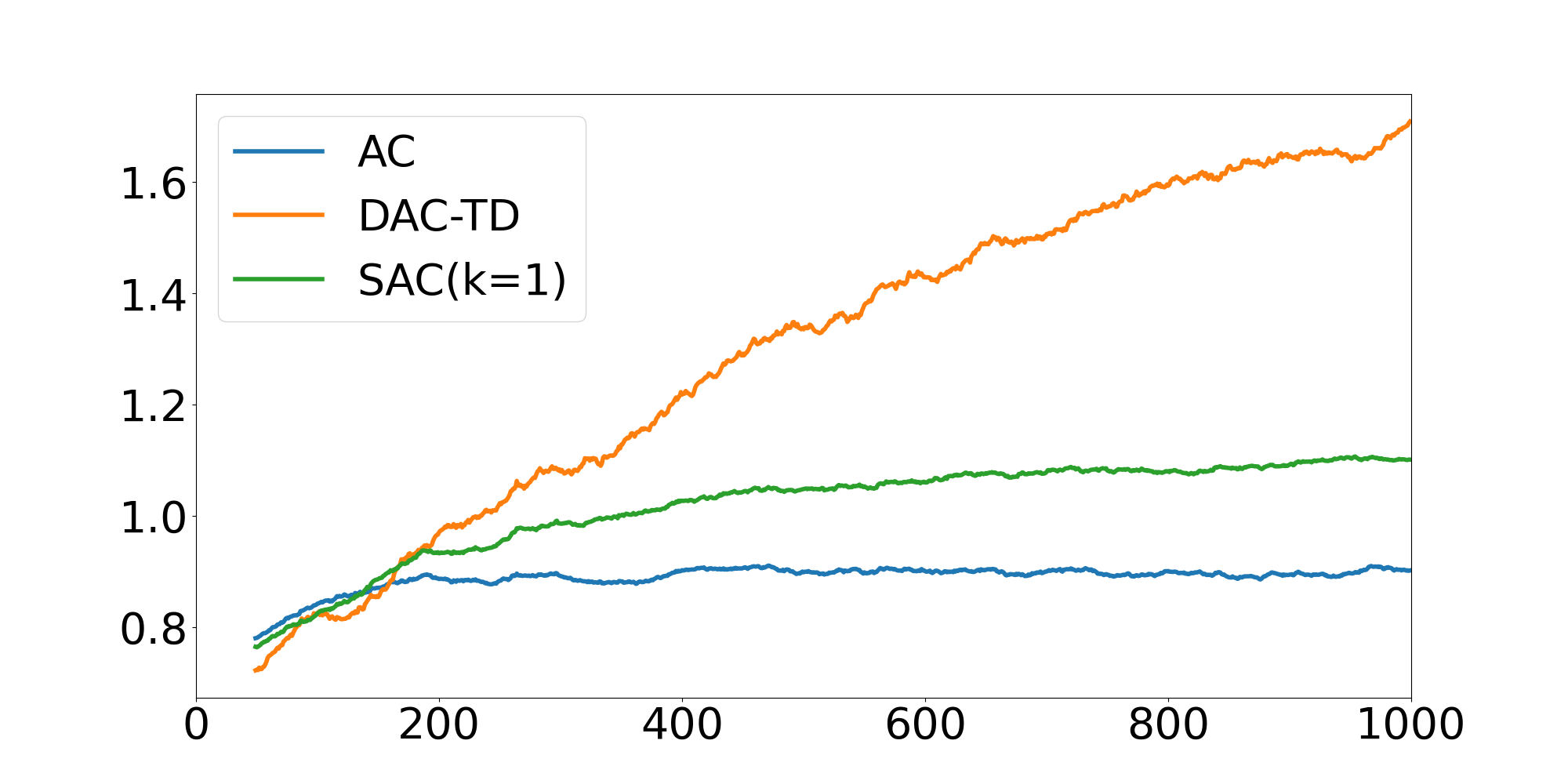}
\caption{Team-average episode returns of the DAC-TD, SAC, and AC agents.}
\label{fig:synthetic}
\end{figure}

\section{Conclusion and future work}\label{sec:conclusion}
In this paper, we introduce a new method for cooperative multi-agent reinforcement learning, where each agent receives private rewards, makes local observations, and communicates with other agents to achieve a team-average objective. Furthermore, we consider that communication channels are subject to time delays and packet dropouts, which is a relevant assumption in the decentralized setting. We developed the decentralized AC algorithm with TD error aggregation to tackle the above-mentioned challenges. We justify the design of this scalable and communication-efficient algorithm through a rigorous mathematical analysis that proves strong convergence to a policy that locally maximizes the team-average objective function. Furthermore, we verify that agents that employ our algorithm are fully cooperative in the empirical analysis.

\bibliography{references}

\begin{appendices}

\section{Stochastic Approximation}\label{appendix:SA}
In this part, we review theoretical results on the stochastic approximation, which provide a framework for convergence proofs in our work. We take under the scope two methods for stochastic approximation. The first method assumes that samples of a random variable (observations) are readily available for the parameter updates, while the second method assumes that the samples are obtained after multiple time steps, which leads to updates with delayed observations. We adopt sufficient conditions for the convergence of algorithms for stochastic approximation from \cite{kushner2003book} and \cite{borkar2009book} and use the notation from the former.
\subsection{Unconstrained stochastic approximation with correlated noise}\label{appendix:SA unc}
In the first part, we analyze the asymptotic behavior of stochastic updates on an irreducible Markov chain with correlated noise. We let $\theta_n$, $Y_n$ and $\xi_n$ denote the estimated parameter vector, observation, and state in a Markov chain defined on a finite set, respectively. We define the filtration $\mathcal{F}_n=\sigma(\theta_0,Y_{i-1},\xi_i,i\leq n)$. The unconstrained stochastic updates are given as follows
\begin{align}
\theta_{n+1}&=\theta_n+\epsilon_n\big[\mathbb{E}(Y_n\vert\mathcal{F}_n)+\delta M_n\big]\nonumber\\
&=\theta_n+\epsilon_n\big[g_n(\theta_n,\xi_n)+\delta M_n+\beta_n\big]\label{alg_unc}
\end{align}
where $\delta M_n=Y_n-E(Y_n\vert\mathcal{F}_n)$ is a martingale difference. 
\begin{assumption}\label{appendix:as1}
We introduce assumptions for the updates in \eqref{alg_unc} in the following lines.
\begin{enumerate}
	
	\item The function $g_n(\cdot,\cdot)$ is Lipschitz continuous in its first argument.
	
	\item $\{\xi_n\}_{n\geq 0}$ is an irreducible Markov chain with stationary distribution $\eta(\xi)$.
	
	\item The step size sequence $\{\epsilon_n\}_{n\geq 0}$ satisfies $\epsilon_n\geq 0$, $\sum_n\epsilon_n=\infty$, and $\sum_n\epsilon_n^2<\infty$.
	
	\item The martingale difference sequence $\{\delta M_n\}_{n\geq 0}$ satisfies for all $n\geq 0$ and some $c>0$ $$\mathbb{E}\big(\Vert\delta M_n\Vert^2\vert\theta_k,\delta M_{k-1},Y_k,k\leq n\big)\leq c\cdot(1+\Vert\theta_n\Vert^2).$$
		
	\item The sequence $\{\beta_n\}_{n\geq 0}$ is a bounded random sequence that satisfies $\lim_{n\rightarrow\infty}\beta_n=0$ with probability one.

\end{enumerate}
\end{assumption}

\begin{theorem}\label{appendix:ODE_unc}
Assumption~\ref{appendix:as1}, the asymptotic behavior of the algorithm in \eqref{alg_unc} is given by an ODE
\begin{align}
\dot\theta=\bar{g}(\theta)=\sum_i\eta(i)g(\theta,i).\label{appendix:ODE}
\end{align}
Furthermore, if $\sup_n\Vert\theta_n\Vert<\infty$ with probability one and the ODE \eqref{appendix:ODE} has a unique globally asymptotically stable equilibrium $\theta^*$, then $\theta_n\rightarrow\theta^*$ as $n\rightarrow\infty$.
\end{theorem}
\begin{theorem}
Under Assumption~\ref{appendix:as1}, suppose that $\lim_{c\rightarrow\infty}\frac{\bar{g}(c\theta)}{c}=g_\infty(\theta)$ exists uniformly on compact sets for some $g_\infty\in C(\mathbb{R}^n)$. If the ODE $\dot\theta=g_\infty(\theta)$ has origin as the unique globally asymptotically stable equilibrium, then $\sup_n\Vert\theta_n\Vert<\infty$ with probability one.\label{appendix:ODE bounded}
\end{theorem}

\subsection{Pipelined constrained stochastic approximation with martingale difference noise}\label{appendix:SA con}
In this part, we analyze the asymptotic behavior of constrained stochastic updates with martingale difference noise. The analytical results for strong convergence of the stochastic approximation with no delays in observations carry over to the stochastic approximation with delayed observations (see \cite[Chapter~12]{kushner2003book}). We let $\theta_n$ and $Y_n$ denote the estimated parameter and observation, respectively. We define the filtration $\mathcal{F}_n=\sigma(\theta_0,Y_{i-K},i\leq n)$, where $K>0$ is a constant delay. The constrained stochastic updates are given as follows
\begin{align}
\theta_{n+1}&=\Psi_\Theta\bigg(\theta_n+\epsilon_n\big[\mathbb{E}(Y_{n-K}\vert\mathcal{F}_{n-K})+\delta M_{n-K}\big]\bigg)\nonumber\\
&=\theta_n+\epsilon_n\big[\bar{g}(\theta_{n-K})+\delta M_{n-K}+\beta_{n-K}+Z_{n-K}\big],\label{alg_con}
\end{align}
where $\Psi_\Theta(\cdot)$ is a projection operator that maps the stochastic updates into a compact admissible set $\Theta$, the term $Z_{n-K}$ is due to the projection, and $\delta M_{n-K}=Y_{n-K}-E(Y_{n-K}\vert\mathcal{F}_{n-K})$ is a martingale difference.
\begin{assumption}\label{appendix:as2}
We introduce assumptions for the constrained algorithm in \eqref{alg_con}.

\begin{enumerate}
	\item The function $\bar{g}(\cdot)$ is continuous.
	
	\item The step size sequence $\{\epsilon_n\}_{n\geq 0}$ satisfies $\epsilon_n\geq 0$, $\sum_n\epsilon_n=\infty$, and $\sum_n\epsilon_n^2<\infty$.
	
	\item The sequence $\{\beta_n\}_{n\geq 0}$ is a bounded random sequence that satisfies $\lim_{n\rightarrow\infty}\beta_n=0$ with probability one.
	
	\item The martingale difference sequence $\{\delta M_n\}_{n\geq 0}$ satisfies for all $n\geq 0$ and some $c>0$\footnote{This is a sufficient condition for the convergence of the martingale difference, which is proved in \cite[p.128]{kushner2003book}} $$\mathbb{E}\big(\Vert\delta M_n\Vert^2\vert\theta_k,\delta M_{k-1},Y_k,k\leq n\big)\leq c\cdot(1+\Vert\theta_n\Vert^2).$$
	
	\item The constraint set $\Theta$ is a hyperrectangle. In other words, for finite vectors $l$ and $u$ such that $l<u$, the constraint set is defined as $\Theta=\{\theta:l\leq\theta\leq u\}$.

\end{enumerate}
\end{assumption}
\begin{theorem}\label{thm:ODE_con}
Under Assumption~\ref{appendix:as2}, the asymptotic behavior of algorithm \eqref{alg_con} is described by an ODE
\begin{align}
\dot\theta=\bar{g}(\theta)+z.\label{ODE2}
\end{align}
\end{theorem}
\begin{theorem}\label{thm:ODE_con 2}
Suppose that the solution of the ODE in \eqref{ODE2} asymptotically converges to a compact set of stationary points $\Theta^*$. Then, the parameters $\theta_n$ in algorithm \eqref{alg_con} converge with probability one to $\Theta^*$ as $n\rightarrow\infty$.
\end{theorem}

\section{Proofs of Convergence Theorems}\label{appendix:proofs}
In this section, we present proofs of Theorem~\ref{thm: critic} and Theorem~\ref{thm:actor}. The proofs of convergence of the critic and the actor are motivated by analytical results on the stochastic approximation that are presented in Appendix~\ref{appendix:SA}.
\subsection{Proof of Theorem~\ref{thm: critic}}
We slightly abuse notation by writing $\phi_t=\phi(s_t)$. We define $A_t^i=\phi_t^i(\gamma\phi_{t+1}^i-\phi_t^i)^T$ and $b_t^i=r_{t+1}^i\phi_t^i$, which allows us to write the critic updates as $v_{t+1}^i=v_t^i+\beta_t(A_t^iv_t^i+b_t^i)$. We define the filtration $\mathcal{F}_{t,v}^i=\sigma(v_0^i,Y_{\tau-1}^i,\xi_\tau,\tau\leq t)$, which consists of the initial parameter value $v_0^i$, the history of the Markov chain $\xi_\tau=(s_\tau,a_\tau)$, and the observations $Y_\tau^i=A_\tau^iv_\tau^i+b_\tau^i$. The filtration $\mathcal{F}_{t,v}^i$ describes the evolution of the parameters $v_t^i$ up to time~$t$. Furthermore, we define a function $g_t^i(v_t^i,\xi_t)=\mathbb{E}_\mathcal{P}\big(A_t^iv_t^i+b_t^i\vert\mathcal{F}_{t,v}^i\big)$ and a martingale difference $\delta M_t^i=A_t^iv_t^i+b_t^i-\mathbb{E}_\mathcal{P}\big(A_t^iv_t^i+b_t^i\vert\mathcal{F}_{t,v}^i\big)$. We rewrite the critic updates as follows
\begin{align}
v_{t+1}^i=v_t^i+\beta_t\big[g_t^i(v_t^i,\xi_t)+\delta M_t^i\big].\label{critic update}
\end{align}
To establish convergence of the recursion in \eqref{critic update}, we check the conditions in Appendix~\ref{appendix:SA unc}. We note that $A_t^i$ and $b_t^i$ are uniformly bounded by Assumption~\ref{as:linear_approx} and \ref{as:reward_bound}. Therefore, the function $g_t^i(v_t^i,\xi_t^i)$ is Lipschitz continuous in its first argument. Furthermore, the same assumptions imply that the martingale difference $\delta M_t^i$ satisfies
\begin{align*}
\mathbb{E}\big(\Vert\delta M_t^i\Vert^2\vert\mathcal{F}_{t,v}^i\big)&=\mathbb{E}\big(\Vert A_t^iv_t^i+b_t^i-\mathbb{E}_\mathcal{P}\big(A_t^iv_t^i+b_t^i\vert\mathcal{F}_{t,v}^i\big)\Vert^2\vert\mathcal{F}_{t,v}^i\big)\\
&\leq K\cdot(1+\Vert v_t^i\Vert^2).
\end{align*}
The Markov chain $\{\xi_t\}_{t\geq 0}$ is irreducible and aperiodic by Assumption~\ref{as:policy}. Its stationary distribution under policy $\pi(a\vert s;\theta)$, denoted as $\eta_\pi(\xi)$, satisfies $\eta_\pi(\xi)=d_\pi(s)\pi(a\vert s)$. Finally, the step size sequence $\beta_t$ satisfies Assumption~\ref{as:step_size}. Hence, the asymptotic behavior of the local critic updates follows an ODE
\begin{align*}
\dot{v}^i&=\bar{g}^i(v^i)=\sum_{\xi} \eta_\pi(\xi)g^i(v^i,\xi)=\sum_{s}d_\pi(s)\sum_{a}\pi(a\vert s)g^i(v^i,s,a)\\
&= \sum_{s}d_\pi(s)\sum_{a}\pi(a\vert s)\sum_{s^\prime} p(s^\prime\vert s,a)\big[A^i(s,s^\prime)v^i+b^i(s,a,s^\prime)\big]\\
&=\sum_{s}d_\pi(s)\bigg(\sum_{s^\prime} p_\pi(s^\prime\vert s)A^i(s,s^\prime)v^i+\sum_{a}\pi(a\vert s)\sum_{s^\prime} p(s^\prime\vert s,a)b^i(s,a,s^\prime)\bigg)\\
&=\Phi^{iT}D_\pi(\gamma P_\pi-I)\Phi^i v^i+\Phi^{iT}D_\pi\hat{R}_\pi^i.
\end{align*}
Furthermore, we consider a function $g_\infty^i(y)=\lim_{c\rightarrow\infty}\frac{\bar{g}^i(cy)}{c}=\Phi_i^TD_\pi(\gamma P_\pi-I)\Phi_i y$ and analyze its eigenvalues. The diagonal matrix $D_\pi$ is positive definite with probability one as it denotes the stationary distribution of states. The eigenvalues of the stochastic matrix $P_\pi$ are less than or equal to one, and hence $\gamma P_\pi-I$ has strictly negative eigenvalues. Denoting $\lambda$ and $x$ an arbitrary eigenvalue-eigenvector pair of the matrix product $D_\pi(\gamma P_\pi-I)$, we obtain
\begin{align*}
D_\pi(\gamma P_\pi-I)x&=\lambda x \implies \lambda=\frac{x^T(\gamma P_\pi-I)^TD_\pi(\gamma P_\pi-I)x}{x^T(\gamma P_\pi-I)x}<0
\end{align*}
since the numerator is positive definite and the denominator is negative definite with probability one. Hence the eigenvalues of $D_\pi(\gamma P_\pi-I)$ are strictly negative with probability one. Using the fact that $\Phi^i$ is a full rank matrix by Assumption~\ref{as:linear_approx} and considering $x\in Im(\Phi^i)$, the eigenvalues of $\Phi_i^TD_\pi(\gamma P_\pi-I)\Phi_i$ are negative with probability one. Therefore, $\dot{y}=g_\infty(y)$ has a unique asymptotically stable equilibrium at the origin with probability one. Applying Theorem~\ref{appendix:ODE bounded}, we obtain $\sup_t\Vert v_t^i\Vert<\infty$ with probability one. Finally, the critic parameters $v_t^i$ converge to the equilibrium of the ODE in \eqref{critic ODE} with probability one by Theorem~\ref{appendix:ODE_unc}.

\subsection{Proof of Theorem~\ref{thm:actor}}
The actor updates are given as $\theta_{t+1}=\theta_t+\alpha_{t-K}\big[\delta_{t-K}\cdot\nabla_{\theta}\log\pi(a_{t-K}\vert s_{t-K};\theta_{t-K})+Z_{t-K}\big]$, where $Z_{t-K}$ is a term that projects the parameters $\theta_t$ into the compact set $\Theta$. We define the filtration $\mathcal{F}_{t,\theta}=\sigma(\theta_0,Y_{\tau-K-1},\tau\leq t)$, which consists of the initial parameter value $\theta_0$ and observations $Y_\tau=\delta_{\tau}\cdot\nabla_{\theta}\log\pi(a_\tau\vert s_\tau;\theta_\tau)+Z_\tau$.  We write the actor updates as follows
\begin{align}
\theta_{t+1}=\theta_t+\alpha_{t-K}\big[\bar{g}(\theta_{t-K})+\delta M_{t-K}+\beta_{t-K}+Z_{t-K}\big],\label{actor update}
\end{align}
where the functions $\bar{g}(\cdot)$, $\delta M_{t-K}$, and $\beta_{t-K}$ are given as
\begin{align*}
\bar{g}(\theta_{t-K})&=\mathbb{E}_{d_{\pi_{t-K}},\pi_{t-K},\mathcal{P}}\big[\delta_{t-K,\pi}\cdot\nabla_{\theta}\log\pi(a_{t-K}\vert s_{t-K};\theta_{t-K})\big]\\
\delta M_{t-K}&=\delta_{t-K}\cdot\nabla_{\theta}\log\pi(a_{t-K}\vert s_{t-K};\theta_{t-K})\\
&\qquad-\mathbb{E}_{d_{\pi_{t-K}},\pi_{t-K},\mathcal{P}}\big[\delta_{t-K}\cdot\nabla_{\theta}\log\pi(a_{t-K}\vert s_{t-K};\theta_{t-K})\vert\mathcal{F}_{t-K,\theta}\big]\\
\beta_{t-K}=&\mathbb{E}_{d_{\pi_{t-K}},\pi_{t-K},\mathcal{P}}\big[(\delta_{t-K}-\delta_{t-K,\pi})\cdot\nabla_{\theta}\log\pi(a_{t-K}\vert s_{t-K};\theta_{t-K})\vert\mathcal{F}_{t-K,\theta}\big]
\end{align*}
Next, we verify conditions in Appendix~\ref{appendix:SA con} to establish convergence of the constrained stochastic updates in \eqref{actor update}. We note that $\delta_{t,\pi}$ is continuous in $d_{\pi}$ and $p_{\pi}$ since they are are continuous in $\theta_t$. Furthermore, the term $\nabla_{\theta}\log\pi(a_t\vert s_t;\theta_t)$ is continuous in $\theta_t$ by Assumption~\ref{as:policy}. Therefore, the function $\bar{g}(\cdot)$ is continuous. The step size sequence satisfies Assumption~\ref{as:step_size}. We note that $\lim_{t\rightarrow\infty}\delta_{t}=\delta_{t,\pi}$ with probability one on the faster timescale, and so $\lim_{t\rightarrow\infty}\beta_t=0$ with probability one. By Assumption~\ref{as:linear_approx} and \ref{as:reward_bound}, the term $\delta_{t}$ is uniformly bounded. Furthermore, the term $\nabla_{\theta}\log\pi(a_{t}\vert s_{t};\theta_{t})$ is uniformly bounded for $\theta_t\in\Theta$ by Assumption~\ref{as:policy} and \ref{as:policy_updates}. Hence, the martingale difference $\delta M_t^i$ satisfies $\mathbb{E}_{d_{\pi_{t-K}},\pi_{t-K},\mathcal{P}}\big(\Vert\delta M_{t-K}\Vert^2\vert\mathcal{F}_{t-K}^i\big)<\infty$ for $t\geq K$. Finally, the constraint set is a hyperrectangle by Assumption~\ref{as:policy_updates}. Using Theorem~\ref{thm:ODE_con} in Appendix~\ref{appendix:SA con}, the asymptotic behavior of the actor updates is described by the ODE 
\begin{align}
\dot\theta_t=\bar{g}(\theta_t)+z_t,\label{ODE actor}
\end{align}
where $z$ is the minimum force needed to ensure that $\theta_t\in\Theta$. We define a surrogate objective function
\begin{align*}
J^\prime(\pi_\theta)=\mathbb{E}_{\pi,d_\pi,\mathcal{P}}\big[\sum_{i\in\mathcal{N}}\big(r^i(s,a,s^\prime)+\gamma V^i(s^{i\prime};v_\pi^i)-V(s^i;v_\pi^i)\big)\big].
\end{align*}
By the policy gradient theorem \cite{sutton2018book}, we obtain $\nabla_{\theta}J^\prime(\pi_{\theta_t})=\bar{g}(\theta_t)$. The rate of change of $J^\prime(\pi_{\theta_t})$ is given as
\begin{align*}
\dot{J}^\prime(\pi_{\theta_t})=\nabla_\theta J^\prime(\pi_{\theta_t})^T\big(\nabla_\theta J^\prime(\pi_{\theta_t})+z_t\big).
\end{align*}
The elements with active constraints satisfy $z_t=-\nabla_\theta J^\prime(\pi_{\theta_t})$ while the elements with inactive constraints satisfy $z_t=0$. Therefore, we obtain $\dot{J}^\prime(\pi_{\theta_t})>0$ if $\nabla_\theta J^\prime(\pi_{\theta_t})+z_t\neq0$ and $\dot{J}^\prime(\pi_{\theta_t})=0$ otherwise, which implies that $\theta$ asymptotically converges to a compact set of stationary points $\Theta^*$ that locally maximize $J^\prime(\pi_{\theta_t})$. Using Theorem~\ref{thm:ODE_con 2}, we establish the asymptotic convergence of $\theta_t$ to a set of stationary points $\Theta^*$ with probability one.

\section{Algorithm extension}\label{appendix:extension}
In this section, we provide the proof of Lemma~\ref{lem:TD aggregation2}.
\subsection{Proof of Lemma~\ref{lem:TD aggregation2}}
For simplicity, we consider the dynamics of the first element of the vectors $\hat\Delta_t^i$, $z_t^i$, and $\rho_t^{ij}$, and we note that its index is increased with time. We define $x_0^i=[\hat\Delta_t^i]_1$ and $\hat{z}_0^i=[z_t^i]_1$. Furthermore, we slightly abuse notation and let $\mathcal{N}^i_k$ denote the set of all $k$-hop neighbors of agent~$i$. We will prove by induction that for $\tau\geq 1$, 
\begin{align}
    x_\tau^{i} = \sum_{0\le k\le \tau}\sum_{l\in\mathcal{N}^i_k}\delta_t^l,\label{eq:xx}\qquad \hat{z}_\tau^{ij} = \sum_{l\in\mathcal{N}^i_\tau\setminus \mathcal{N}^j_{\tau-1}} \delta_t^l, \qquad j\in\mathcal{N}^i_1.
\end{align}
Suppose that $\tau=1$. Since $x_0^i = z_{0}^{ij}= \delta_t^i$, we have
\begin{align*}
    x_{1}^i &= x_0^i + \sum_{j\in\mathcal{N}^i_1} x_0^j = \delta_t^i + \sum_{j\in\mathcal{N}^i_1} \delta_t^j = \sum_{0\le k\le 1}\sum_{l\in\mathcal{N}^i_k}\delta_t^l, \\
    \hat{z}_{1}^{ij} &= x_{1}^i - x_{0}^i - x_{0}^j = \sum_{l\in\mathcal{N}^i_1} x_0^l - x_{0}^j = \sum_{l\in\mathcal{N}^i_1 \setminus \{j\}} \delta_t^l =\sum_{l\in\mathcal{N}^i_1 \setminus \mathcal{N}^j_0} \delta_t^l, \qquad j\in\mathcal{N}^i_1.
\end{align*}
Next, suppose that \eqref{eq:xx} hold for $\tau>1$. We have
\begin{align}
    &x_{\tau+1}^i
    = x_{\tau}^i + \sum_{j\in\mathcal{N}^i_1} \Big(x_{\tau}^j - x_{\tau-1}^j-\hat{z}_{\tau-1}^{ij}\Big)\nonumber\\
    &= \sum_{0\le k\le {\tau}}\sum_{l\in\mathcal{N}_k^i}\delta_t^l + \sum_{j\in\mathcal{N}^i_1} \Big(\sum_{0\le k\le {\tau}}\sum_{l\in\mathcal{N}^j_k}\delta_t^l - \sum_{0\le k\le {\tau-1}}\sum_{l\in\mathcal{N}_k^j}\delta_t^l-\sum_{l\in\mathcal{N}^i_{\tau-1}\setminus \mathcal{N}^j_{\tau-2}} \delta_t^l\Big)\nonumber\\
    &= \sum_{0\le k\le {\tau}}\sum_{l\in\mathcal{N}^i_k}\delta_t^l + \sum_{j\in\mathcal{N}^i_1} \Big(\sum_{l\in\mathcal{N}^j_{\tau}}\delta_t^l-\sum_{l\in \mathcal{N}^i_{\tau-1}\setminus \mathcal{N}^j_{\tau-2}}\delta_t^l\Big)\label{eq:proof_x}\\
    &\hat{z}_{\tau+1}^{ij}
    = \hat{z}_{\tau-1}^{ij} + \Big(x_{\tau+1}^i - x_{\tau}^i\Big) - \Big(x_{\tau}^j - x_{\tau-1}^j\Big)\nonumber\\
    &= \sum_{l\in\mathcal{N}^i_{\tau-1 }\setminus \mathcal{N}^j_{{\tau}-2}} \delta_t^l + \sum_{l\in\mathcal{N}^i_{\tau+1}}\delta_t^l - \sum_{l\in\mathcal{N}^j_{\tau}}\delta_t^l, \qquad j\in\mathcal{N}^i_1.\label{eq:proof_y}
\end{align}
Since $\mathcal{G}$ is an undirected acyclic graph, we have $\mathcal{N}^j_{\tau} \setminus \mathcal{N}^i_{\tau-1}=\mathcal{N}^j_{\tau} \cap \mathcal{N}^i_{\tau+1}$, $\mathcal{N}^j_{\tau} \setminus \mathcal{N}^i_{\tau+1}=\mathcal{N}^j_{\tau} \cap \mathcal{N}^i_{\tau-1}$, and $\mathcal{N}^i_{\tau-1}\setminus \mathcal{N}^j_{\tau-2}\subset\mathcal{N}^j_{\tau}$  for any two immediate neighbors $i$ and $j$. We use this fact to obtain
\begin{align*}
    \mathcal{N}^j_{\tau} \setminus (\mathcal{N}^i_{\tau-1}\setminus\mathcal{N}^j_{{\tau}-2})
    &=(\mathcal{N}^j_{{\tau}-2}\cap\mathcal{N}^j_{\tau})\cup  (\mathcal{N}^j_{\tau} \setminus \mathcal{N}^i_{\tau-1})\\
    &= \mathcal{N}^j_{\tau} \setminus \mathcal{N}^i_{\tau-1}\\
    &= \mathcal{N}^j_{\tau} \cap \mathcal{N}^i_{\tau+1}.\label{eq:proof_x_set}
 \end{align*}
Applying the relationships between the sets to \eqref{eq:proof_x} and \eqref{eq:proof_y}, we obtain
\begin{align*}
    x_{\tau+1}^i 
    &= \sum_{0\le k\le {\tau}}\sum_{l\in\mathcal{N}^i_k}\delta_t^l + \sum_{j\in\mathcal{N}^i_1} \sum_{l\in \mathcal{N}^j_{\tau} \cap \mathcal{N}^i_{\tau+1}}\delta_t^l\\
    &= \sum_{0\le k\le {\tau}}\sum_{l\in\mathcal{N}^i_k}\delta_t^l + \sum_{l \in \mathcal{N}^i_{\tau+1}}\delta_t^l\\
    &= \sum_{0\le k\le {\tau}+1}\sum_{l\in\mathcal{N}^i_k}\delta_t^l,\\
    \hat{z}_{\tau+1}^{ij}
    &=\sum_{l\in\mathcal{N}^i_{\tau-1 }\cap \mathcal{N}^j_{{\tau}}} \delta_t^l + \sum_{l\in\mathcal{N}^i_{\tau+1}}\delta_t^l - \sum_{l\in\mathcal{N}^j_{\tau}}\delta_t^l\\
    &=\sum_{l\in\mathcal{N}^i_{\tau+1}}\delta_t^l - \sum_{l\in\mathcal{N}^j_{\tau}\cap\mathcal{N}^i_{\tau+1 }}\delta_t^l\\
    &=\sum_{l\in\mathcal{N}^j_{\tau+2}\cap\mathcal{N}^i_{\tau+1 }}\delta_t^l\\
    &=\sum_{l\in\mathcal{N}^i_{\tau+1}\setminus\mathcal{N}^j_{\tau}}\delta_t^l
\end{align*}
Since the set $\mathcal{N}^i_k$ is empty for $k>K$, it holds that $x_\tau^1=x_\tau^2=\dots=x_\tau^N=\sum_{j\in\mathcal{N}} \delta_t^j$ for $\tau\geq K$.

\end{appendices}

\end{document}